\documentclass[3p,doublespacing,reviewcopy,11pt]{elsarticle}
\usepackage{ecrc}
\usepackage{pifont}
\usepackage{amsmath, amsthm, amssymb, amsfonts}
\usepackage{graphicx}
\usepackage{geometry}
\usepackage{txfonts}
\usepackage{natbib}
\usepackage[T1]{fontenc}
\usepackage[all]{xy}
\usepackage{times}
\usepackage{yfonts}
\usepackage[utopia]{mathdesign}

\theoremstyle{plain}
\newtheorem{Theorem}{Theorem}

\newtheorem{Lemma}{Lemma}

\newtheorem{Proposition}{Proposition}
\newtheorem{Corollary}{Corollary}

\theoremstyle{definition}
\newtheorem{Definition}{Definition}[section]
\theoremstyle{remark}
\newtheorem{Remark}{Remark}

\volume{??}

\firstpage{1} \journal{Finite Fields and Their Applications}

\usepackage{amssymb}

\usepackage[figuresright]{rotating}

\begin{document}

\begin{frontmatter}

\title{Galois Correspondence on Linear Codes over Finite Chain Rings}

\author[aft]{A. Fotue Tabue}\ead{alexfotue@gmail.com}
\author[ed]{E. Martínez-Moro\corref{cor1}} \ead{edgar@maf.uva.es}\cortext[cor1]{Corresponding author. Partially funded by Spanish MICINN  under grant MTM2015-665764-C3-1-P}
\author[cm]{C. Mouaha}\ead{cmouaha@yahoo.fr}

\address[aft]{Department of mathematics, Faculty of Sciences,  University of Yaoundé 1, Cameroon}
\address[ed]{Institute of Mathematics, University of Valladolid, Spain}
\address[cm]{Department of mathematics,  Higher Teachers Training College of Yaoundé, University of Yaoundé 1, Cameroon}

\begin{abstract}

Given $\texttt{S}|\texttt{R}$ a finite Galois extension
of finite chain rings and $\mathcal{B}$ an $\texttt{S}$-linear
code we define two Galois operators, the closure
operator and the interior operator.  We proof that a linear code
is Galois invariant if and only if the row standard form of its
generator matrix has all entries in the fixed ring by the Galois
group and show a Galois correspondence in the class of $\texttt{S}$-linear
codes. As applications some improvements of  upper  and lower bounds for the rank of the restriction and trace code are given  and
some applications to $\texttt{S}$-linear cyclic codes are shown.

\end{abstract}


\end{frontmatter}

\section{Introduction}\label{sec:lay}

Let $\texttt{R}$ a be finite chain ring of index nilpotency $s$,
$\texttt{S}$ the \emph{Galois extension} of $\texttt{R}$
of rank $m$, and $G$ the group of ring automorphisms of
$\texttt{S}$ fixing $\texttt{R}$. We will denote by
$\mathcal{L}(\texttt{S}^\ell)$ (resp.
$\mathcal{L}(\texttt{R}^\ell)$) the set  of $\texttt{S}$-linear
codes (resp. $\texttt{R}$-linear codes) of length $\ell.$ There
are two classical constructions that allow us to build an element
of $\mathcal{L}(\texttt{R}^\ell)$ from an element $\mathcal{B}$ of
$\mathcal{L}(\texttt{S}^\ell)$. One is the
 \emph{restriction code} of $\mathcal{\mathcal{B}}$   which is defined as
$\texttt{Res}_\texttt{R}(\mathcal{B}):=\mathcal{B}\cap
\texttt{R}^{\ell}.$ The second one is based on the fact that the
trace map
$\texttt{Tr}_\texttt{R}^\texttt{S}=\sum\limits_{\sigma\in
G}\sigma$ is a linear form,
therefore it follows that the set
\begin{align}\label{trace}
\texttt{Tr}_\texttt{R}^\texttt{S}(\mathcal{B}):=\left\{(\texttt{Tr}_\texttt{R}^\texttt{S}(c_1),\cdots,\texttt{Tr}_\texttt{R}^\texttt{S}(c_\ell))\,|\,(c_1,\cdots,c_\ell)\in\mathcal{B}\right\},
\end{align}
 is an $\texttt{R}$-linear code. The relation between the trace code and the restriction code will be given by a generalization (see Theorem~\ref{Delsarte}) of
the celebrated result due to Delsarte \cite{Del75}
\begin{align}\label{del}
\texttt{Tr}_\texttt{R}^\texttt{S}(\mathcal{B}^{\perp_{\varphi'}})=\texttt{Res}_\texttt{R}(\mathcal{B})^{\perp_\varphi},
\end{align}
where $\perp_\varphi$ and $\perp_{\varphi'}$ denote the duality
operators associated to the bilinear forms
$\varphi:\texttt{R}^\ell\times\texttt{R}^\ell\rightarrow\texttt{R}$
and
$\varphi':\texttt{S}^\ell\times\texttt{S}^\ell\rightarrow\texttt{S}$
 defined in Section~\ref{sec:bilinear}.

Restriction codes is a core topic in coding theory. Note that many well-known codes can be defined as a restriction code, for
instance BCH codes and, more generally, alternant codes (see
\cite[Chap.12]{WS77}). Also in  \cite{Bie02} restriction codes and
the closure operator on the set of linear codes over
$\mathbb{F}_{q^m}$ of length $\ell$ are used intensely to
determine the parameters of   additive cyclic codes. More
recently, the restriction code, trace code and Galois invariance
over  extension of finite fields were studied in \cite{MA10}, and
their results were extended to separable Galois extensions of
finite chain rings in \cite{MNR13}. In this paper, we will study
the following operators on
 $\mathcal{L}(\texttt{S}^\ell)$

\[\begin{array}{cccc}
  ^{\widetilde{~~~}} : & \mathcal{L}(\texttt{S}^\ell) & \rightarrow & \mathcal{L}(\texttt{S}^\ell) \\
    & \mathcal{B} & \mapsto & \widetilde{\mathcal{B}}:=\underset{\sigma\in G}{\bigvee
    }\sigma(\mathcal{B}),
\end{array}\textrm{~~~~~~ and ~~~~~~~ }
\begin{array}{cccc}
  ^{\overset{\circ}{~~~~}} : & \mathcal{L}(\texttt{S}^\ell) & \rightarrow & \mathcal{L}(\texttt{S}^\ell) \\
    & \mathcal{B} & \mapsto & \overset{\circ}{\mathcal{B}}:=\underset{\sigma\in G}{\bigcap}\sigma(\mathcal{B})
\end{array}, \]
and we will give answer to the question  if
there  is a Galois correspondence between $\mathcal{B}$ and $G$ for each  $\mathcal{B}$ in $\mathcal{L}(\texttt{S}^\ell)$.
We will make  some improvements of
the bounds for  the rank of restriction and trace of  $\texttt{S}$-linear code and
also,
    when $\ell$ and $q$ are coprime, we will answer whether a linear cyclic code over $\texttt{R}$   is  a restriction of a linear cyclic Galois invariant code over a extension  of
    $\texttt{R}$  or not.

    The outline of the paper is as follows. In Section 2 we give some preliminaries in finite commutative chain rings and their  Galois extensions. We also  formulate a generator matrix in row standard form  for
linear codes over finite chain rings. Section 3 presents the study of the
Galois operators on the lattice of linear codes. Finally in
Section 4 we describe
 linear cyclic codes over finite chain rings as the restriction
of a linear cyclic code over a   Galois extension of a
finite chain ring.

\section{Preliminaries}\label{sec:pre}

\subsection{Finite chain rings}

A finite commutative ring with identity is called a \emph{finite
chain ring} if its ideals are linearly ordered by inclusion. It is
well known that every ideal of a finite chain ring is principal and therefore
its maximal ideal is unique.  $\texttt{R}$  will denote a finite
chain ring, $\theta$ a generator of its maximal ideal
$\textgoth{m}=\texttt{R}\theta$,
$\mathbb{F}_{p^n}=\texttt{R}/\textgoth{m}$ its residue field, and $\pi :
\texttt{R}\rightarrow \mathbb{F}_q$ the canonical projection. As stated before the
ideals of $\texttt{R}$ form a chain $ \texttt{R}\supsetneq
\texttt{R}\theta \supsetneq \cdots\supsetneq
\texttt{R}\theta^{s-1} \supsetneq \texttt{R}\theta^s =\{0\}$  where the
integer $s$ is called the \emph{nilpotency index} of $\texttt{R}.$
It is easy to see that the cardinal of $\texttt{R}^\times$, the set of ring units,
  is $p^{n(s-1)}(p^n-1).$ Thus
$\texttt{R}^\times\simeq\Gamma(\texttt{R})^*\times(1+\texttt{R}\theta)$
where  $\Gamma(\texttt{R})^*=\{ b\in \texttt{R}\mid b\neq 0,\,
b^{p^n}=b \}$ is the only  subgroup of $\texttt{R}^\times$
isomorphic to $\mathbb{F}_{p^n}\setminus\{0\}$. The set
$\Gamma(\texttt{R})=\Gamma(\texttt{R})^*\cup\{0\}$ is a complete
set of representatives of $\texttt{R}$ modulo $\theta$ and it is
called the \emph{Teichmüller set} of $\texttt{R}.$
The set
$\Gamma(\texttt{R})$ is a coordinate set of $\texttt{R}$
\cite[Proposition 3.3]{Nechaev}, i.e. each element $a \in
\texttt{R}$ can be expressed uniquely as a \emph{$\theta$-adic
decomposition}
\begin{align}\label{adic}a = a_0 + a_1\theta+\cdots+a_{s-1}\theta^{s-1},\end{align}
where $a_0, a_1, \cdots, a_{s-1}\in\Gamma(\texttt{R})$. The
$\theta$-adic decomposition of elements of $\texttt{R}$ allows us
to defines the $t$-th $\theta$-adic coordinate function as
\begin{align}\label{tepr}
\begin{array}{cccc}
 \gamma_{\textgoth{t}}: & \texttt{R} & \rightarrow & \Gamma(\texttt{R}) \\
    & a & \mapsto & a_t
  \end{array}\quad \textgoth{t}=0,1,\cdots,s-1,
\end{align}
where $a= \gamma_{0}(a)+ \gamma_{1}(a)\theta+\cdots+
\gamma_{s-1}(a)\theta^{s-1}.$ Therefore we have a \emph{valuation
function} of $\texttt{R},$ defined by
$\vartheta_\texttt{R}(a):=\texttt{min}\{t\in\{0,1,\cdots,s\}\,|\,\gamma_{\textgoth{t}}(a)\neq
0\}$ and a \emph{degree function} of $\texttt{R},$ defined by
$\texttt{deg}_\texttt{R}(a):=\texttt{max}\{t\in\{0,1,\cdots,s\}\,|\,\gamma_{\textgoth{t}}(a)\neq
0\},$ for each $a$ in $\texttt{R}.$ We will assume that
$\vartheta_\texttt{R}(0)=s$ and
$\texttt{deg}_\texttt{R}(0)=-\infty.$

\subsection{Galois extensions}

Let $\texttt{R}$ and $\texttt{S}$ be two finite chain rings with
residue fields $\mathbb{F}_q$ and $\mathbb{F}_{q^m}$ respectively.
We say that $\texttt{S}$ is an \emph{extension} of $\texttt{R}$
and we denote it by $\texttt{S}|\texttt{R}$ if
$\texttt{R}\subseteq \texttt{S}$ and $1_\texttt{R} =
1_\texttt{S}.$ $\texttt{Aut}_\texttt{R}(\texttt{S})$  will denote
the group of automorphisms of $\texttt{S}$ which fix the elements
of $\texttt{R}$.  Note that the map $\sigma:a\mapsto
\sum\limits_{\textgoth{t}=0}^{s-1}\gamma_{\textgoth{t}}(a)^q\theta^\textgoth{t}
$ for all $a\in\texttt{S},$ is in
$\texttt{Aut}_\texttt{R}(\texttt{S})$ and throughout of this paper
$G$ will be  the subgroup of $\texttt{Aut}_\texttt{R}(\texttt{S})$
generated by $\sigma$. For each subring $\texttt{T}$ such that
$\texttt{R}\subseteq\texttt{T}\subseteq\texttt{S}$
  and for each subgroup of $G$ one can respectively define the
\emph{fixed group} of $\texttt{T}$ in $G$  and the \emph{fixed
ring} of $H$ in $\texttt{S}$ as
$$
 \texttt{Stab}_G(\texttt{T}):= \biggl\{\varrho\in G\,\biggr|\, \varrho(a) = a,\text{ for all } a\in
 \texttt{T}\biggr\},
\qquad \texttt{Fix}_\texttt{S}(H) := \biggl\{a\in \texttt{S }\,\biggr|\,
\varrho(a) = a,\text{ for all } \varrho\in H\biggr\}.$$

\begin{Definition}\label{Galois}
The ring $\texttt{S}$ is a \emph{Galois extension} of $\texttt{R}$
with Galois group $G$ if
\begin{enumerate}
    \item $\texttt{Fix}_\texttt{S}(G)=\texttt{R}$ and
    \item there are elements $\alpha_0,\alpha_1,\cdots,\alpha_{m-1};\alpha^*_0,\alpha^*_1,\cdots,\alpha^*_{m-1}$ in $\texttt{S}$ such that
$$\sum\limits_{t=0}^{m-1}\sigma^i(\alpha_t)\sigma^j(\alpha_t^*)=\delta_{i,j},$$
for all $i,j=0,1,\cdots,|G|-1$(where $\delta_{i,j}=1_\texttt{S}$
if $i=j,$ and $0_\texttt{S}$ otherwise).
\end{enumerate}
\end{Definition}

Note that a Galois extension is separable but the converse is not true in general as it was stated in \cite{{McD74}}, for a complete discussion on this fact  see \cite{Whe92}. If $\texttt{S}|\texttt{R}$ is a \emph{Galois
extension} with Galois group $G$ then
$\texttt{Tr}_\texttt{R}^\texttt{S}$ is a free generator of
$\texttt{Hom}_\texttt{R}(\texttt{S},\texttt{R})$  as an
$\texttt{S}$-module. The following result  in \cite[Chap. III, Theorem 1.1]{DI71} provide us the  Galois correspondence for finite chain
rings.

\begin{Lemma}\label{subG}
Let $\texttt{S}|\texttt{R}$ be a Galois extension with Galois
group $G.$ If $\texttt{T}$ a Galois extension of $\texttt{R}$ and
$\texttt{T}$ is a subring of $\texttt{S},$ then  the Galois group
of $\texttt{T}|\texttt{R}$ is $\texttt{Stab}_G(\texttt{T}).$
Furthermore, $\texttt{Stab}_G(\texttt{Fix}_\texttt{S}(H))=H$ and
$\texttt{Fix}_\texttt{S}(\texttt{Stab}_G(\texttt{T}))=\texttt{T}.$
\end{Lemma}

We say that the pair $(\texttt{Stab}_G,\texttt{Fix}_\texttt{S})$
is a \emph{Galois correspondence} between $G$ and $\texttt{S}.$  Note that the Galois extension
$\texttt{S}|\texttt{R}$ is a free $\texttt{R}$-module with
$|G|=\texttt{rank}_\texttt{R}(\texttt{S})$ (see \cite[Chap. III]{DI71}, \cite[Theorem V.4]{McD74})  and $G=\texttt{Aut}_\texttt{R}(\texttt{S})$
(see \cite[Theorem XV.10]{McD74}). From now on
$\underline{\alpha}:=\{\alpha_0,\alpha_1,\cdots,\alpha_{m-1}\}$
will denote a free $\texttt{R}$-basis of $\texttt{S}$ and
$\mathbb{M}_{\underline{\alpha}}:=\left(\texttt{Tr}_\texttt{R}^\texttt{S}(\alpha_i\alpha_j)\right)_{\substack{0\leq
i<m \\ 0\leq j<m}}$ will be the matrix associate to $\underline{\alpha}.$

\begin{Proposition}\label{Gal1}
Let $\texttt{S}|\texttt{R}$ be a Galois extension with Galois
group $G.$ Then the matrix
$\mathbb{M}_{\underline{\alpha}}$ is
invertible.
\end{Proposition}

\begin{proof} Since $\texttt{rank}_\texttt{R}(\texttt{S})=\texttt{rank}_{\mathbb{F}_q}(\mathbb{F}_{q^m})=m,$
by \cite[Theorem V.5]{McD74},
$\{\pi(\alpha_0),\pi(\alpha_1),\cdots,\pi(\alpha_{m-1})\}$ is an
$\mathbb{F}_q$-basis of $\mathbb{F}_{q^m}.$ So  In fact,
$\texttt{Tr}_{\mathbb{F}_{q}}^{\mathbb{F}_{q^m}}\circ\pi=\pi\circ\texttt{Tr}_\texttt{R}^\texttt{S}$
and
$\pi(\texttt{S}^\times)=\mathbb{F}_{q^m}\backslash\{0\}.$
So by \cite[Theorem 8.3]{Zxw03} the determinant of
$\mathbb{M}_{\underline{\alpha}}$ is a ring unit.
\end{proof}

Hence there are elements  $\alpha_0^*,\alpha_1^*,\cdots,\alpha_{m-1}^*$
in $\texttt{S}$ such that
$\texttt{Tr}_\texttt{R}^\texttt{S}(\alpha_i\alpha^*_j)=\delta_{i,j}$
 and
$(\alpha_0^*,\alpha_1^*,\cdots,\alpha_{m-1}^*)=\mathbb{M}_{\underline{\alpha}}^{-1}(\alpha_0,\alpha_1,\cdots,\alpha_{m-1}).$
Thus the set $\{\alpha_0^*,\alpha_1^*,\cdots,\alpha_{m-1}^*\}$ is a free
$\texttt{R}$-basis of $\texttt{S}$ called the \emph{trace-dual
basis} of $\{\alpha_0,\alpha_1,\cdots,\alpha_{m-1}\}.$

\subsection{Bilinear forms}\label{sec:bilinear}

An \emph{$\texttt{S}$-linear code} of
length $\ell$  is a  $\texttt{S}$-module of $\texttt{S}^{\ell}$,
and the elements of $\mathcal{B}$ are called \emph{codewords}. From now on we will assume that all codes are of length $\ell$ unless   stated otherwise.

\noindent Let be $\textbf{a}$ and $\textbf{b}$ in $\texttt{S}^\ell$, their Euclidean inner product    is defined as
$(\textbf{a},\textbf{b})_\texttt{E}=a_1b_1+a_2b_2+\cdots+a_\ell
b_\ell,$ and if $m$ is even   their Hermitian inner
product   is defined as
$(\textbf{a},\textbf{b})_\texttt{H}=(\sigma^{\frac{m}{2}}(\textbf{a}),\textbf{b})_\texttt{E}$.  Note that
$(-,-)_\texttt{E}$ is a symmetric bilinear form.

\noindent For all $\textbf{a}$ in $\texttt{S}^\ell$ and $\textbf{b}$ in
$\texttt{R}^\ell,$
$\texttt{Tr}_\texttt{R}^\texttt{S}\left((\textbf{a},\textbf{b})_\texttt{E}\right)=\left(\texttt{Tr}_\texttt{R}^\texttt{S}(\textbf{a}),\textbf{b}\right)_\texttt{E},$
and if $m$ is even,
$\texttt{Tr}_\texttt{R}^\texttt{S}\left((\textbf{a},\textbf{b})_\texttt{H}\right)=\texttt{Tr}_\texttt{R}^\texttt{S}\left((\textbf{a},\textbf{b})_\texttt{E}\right),$
since
$\texttt{Tr}_\texttt{R}^\texttt{S}\left(\sigma^{\frac{m}{2}}(\textbf{a})\right)=\texttt{Tr}_\texttt{R}^\texttt{S}(\textbf{a}).$
Throughout the paper $\varphi=(-,-)_\texttt{E}$ and if $m$ is even
$\varphi'=(-,-)_\texttt{H},$ otherwise
$\varphi'=(-,-)_\texttt{E}.$ It is clear that
\begin{align}\label{rela1}
\varphi(\textbf{b},\texttt{Tr}_\texttt{R}^\texttt{S}(\textbf{a}))=\varphi(\texttt{Tr}_\texttt{R}^\texttt{S}(\textbf{a}),\textbf{b})=\texttt{Tr}_\texttt{R}^\texttt{S}(\varphi'(\textbf{a},\textbf{b})),
\hbox{ for all } \textbf{a}\in\texttt{S}^\ell\hbox{ and }\textbf{b}\in\texttt{R}^\ell.
\end{align}

\begin{Lemma}\label{dual-dual} Let $\mathcal{B}$ be an $\texttt{S}$-linear code.
Then
$$\mathcal{B}^{\perp_{\varphi'}}:=\left\{\textbf{a}\in\texttt{S}^\ell\,|\,\varphi'(\textbf{a},\textbf{c})=0,\,\text{
for all }\textbf{c}\in\mathcal{B}\right\}$$ is an
$\texttt{S}$-linear code of the same  length as $\mathcal{B}$ and
$\left(\mathcal{B}^{\perp_{\varphi'}}\right)^{\perp_{\varphi'}}=\mathcal{B}.$

\end{Lemma}

\begin{proof} Let $\textbf{a}$ be in $\texttt{S}^\ell,$  the map
$\varphi'_{\textbf{a}}:=\varphi'(\textbf{a},-)$ is an
$\texttt{S}$-linear form. Therefore
$\mathcal{B}^{\perp_{\varphi'}}=\underset{\textbf{c}\in\mathcal{B}}{\cap
}\texttt{Ker}(\varphi'_\textbf{c})$ is an $\texttt{S}$-linear code. If $m$ is odd,
$\varphi=\varphi'=(-,-)_\texttt{E},$ by \cite[Theorem 3.10.
(iii)]{NS00},
$\left(\mathcal{B}^{\perp}\right)^{\perp}=\mathcal{B}.$ Otherwise,
$$\mathcal{B}^{\perp_\texttt{H}}=\left\{\textbf{a}\in\texttt{S}^\ell\,|\,(\textbf{a},\textbf{c})_{\texttt{E}}=0,\,\text{
for all
}\textbf{c}\in\sigma^{\frac{m}{2}}(\mathcal{B})\right\}=(\sigma^{\frac{m}{2}}(\mathcal{B}))^{\perp}.$$
Thus
$\left(\mathcal{B}^{\perp_{\texttt{H}}}\right)^{\perp_{\texttt{H}}}=\left(\sigma^{\frac{m}{2}}(\mathcal{B})^{\perp}\right)^{\perp_{\texttt{H}}}=\left(\sigma^{\frac{2m}{2}}(\mathcal{B})^{\perp}\right)^{\perp}.$
Since $\sigma^{\frac{2m}{2}}=\sigma^m=\texttt{Id}$ it follows that
$\left(\mathcal{B}^{\perp_{\texttt{H}}}\right)^{\perp_{\texttt{H}}}=\left(\mathcal{B}^{\perp}\right)^{\perp}=\mathcal{B}$.
\end{proof}

The $\texttt{S}$-linear code $\mathcal{B}^{\perp_{\varphi'}}$ is called \emph{$\varphi'$-dual code} of
the code $\mathcal{B}$
associated to the $\texttt{S}$-bilinear form
$\varphi'.$ The following theorem generalizes Delsarte's celebrated
result~\cite{Del75}.

\begin{Theorem}[Delsarte Theorem]\label{Delsarte}

 Let $\mathcal{B}$ be an $\texttt{S}$-linear
code then
$\texttt{Tr}_\texttt{R}^\texttt{S}(\mathcal{B}^{\perp_{\varphi'}})=\texttt{Res}_\texttt{R}(\mathcal{B})^{\perp}.$
\end{Theorem}

\begin{proof}
 Let
$\textbf{a}\in\texttt{Tr}_\texttt{R}^\texttt{S}(\mathcal{B}^{\perp_{\varphi'}}).$
Then $\textbf{a}=\texttt{Tr}_\texttt{R}^\texttt{S}(\textbf{b})$
and $\textbf{b}\in\mathcal{B}^{\perp_{\varphi'}}.$ For all
$\textbf{c}\in\texttt{Res}_\texttt{R}(\mathcal{B}),$
\begin{eqnarray*}
  \varphi(\textbf{a},\textbf{c}) &=& \varphi(\texttt{Tr}_\texttt{R}^\texttt{S}(\textbf{b}),\textbf{c}), \\
    &=&
    \texttt{Tr}_\texttt{R}^\texttt{S}(\varphi'(\textbf{b},\textbf{c})), \text{ from } (\ref{rela1})\\
    &=&\texttt{Tr}_\texttt{R}^\texttt{S}(0),\\
    &=&0.
\end{eqnarray*}
Thus
$\texttt{Tr}_\texttt{R}^\texttt{S}(\mathcal{B}^{\perp_{\varphi'}})\subseteq\texttt{Res}_\texttt{R}(\mathcal{B})^{\perp}.$

    On the other hand, let
    $\textbf{a}\in\left(\texttt{Tr}_\texttt{R}^\texttt{S}(\mathcal{B}^{\perp_{\varphi'}})\right)^{\perp}$, we have that
    for all $\textbf{c}$ in $\mathcal{B}^{\perp_{\varphi'}},$
    $\varphi(\textbf{a},\texttt{Tr}_\texttt{R}^\texttt{S}(\textbf{c}))=0.$
    Thus for all $\lambda$ in $\texttt{S}$
    $\lambda\textbf{c}\in\mathcal{B}$ and  it follows that
    $\varphi(\textbf{a},\texttt{Tr}_\texttt{R}^\texttt{S}(\lambda\textbf{c}))=0$.
    Considering  the relation (\ref{rela1}) we have that
    $\texttt{Tr}_\texttt{R}^\texttt{S}(\varphi'(\textbf{a},\lambda\textbf{c}))=0$
  and
    since $\varphi'(\textbf{a},-)$ is linear  we get that $\varphi'(\textbf{a},\lambda\textbf{c})=\lambda\varphi'(\textbf{a},\textbf{c}).$
Hence,
$\texttt{Tr}_\texttt{R}^\texttt{S}(\lambda\varphi'(\textbf{a},\textbf{c}))=0$
for all $\lambda$ in $\texttt{S}.$ Therefore by
{Lemma}~\ref{dual-dual} we have
$\texttt{Res}_\texttt{R}(\mathcal{B})\subseteq\texttt{Tr}_\texttt{R}^\texttt{S}(\mathcal{B}^{\perp_{\varphi'}}).$

\end{proof}

\subsection{Generator matrix in row standard form}

Usually  to each $\texttt{S}$-linear code can be associated  a generator matrix in standard
form that involves permuting the columns of the original generator matrix, i.e. the code generated by the new matrix is permutation equivalent to the original one (see \cite[Proposition 3.2]{NS00}). We will now reformulate a
more detailed version of this result in order to define a unique
generator matrix in row standard canonical for an
$\texttt{S}$-linear code that generates the code and that will be helpful in determining whether a code is Galois invariant or not.

Let $\textbf{v}_1, \textbf{v}_2, \ldots , \textbf{v}_k$ be
non-zero vectors in $\texttt{S}^\ell$, we say that they  are
\emph{$\texttt{S}$-independent}
 if for all $a_1, a_2,\ldots, a_k$ in $\texttt{S}$ we have that
$a_1\textbf{v}_1+ a_2\textbf{v}_2+ \cdots +
a_k\textbf{v}_k=\textbf{0}$ implies that $a_i\textbf{v}_i
=\textbf{0},$ for all $i.$ Let $\mathcal{B}$ be an
$\texttt{S}$-linear code, the codewords $\textbf{c}_1,
\textbf{c}_2, \ldots , \textbf{c}_k\in \mathcal{B}$ form an
$\texttt{S}$-\emph{basis} of $\mathcal{B}$ if they are independent
and  they \emph{generate} $\mathcal{B}$ as an $\texttt{S}$ module.
Any $\texttt{S}$-linear code $\mathcal{B}$ admits an
$\texttt{S}$-basis and any two $\texttt{S}$-basis of $\mathcal{B}$
has the same number of codewords, see \cite[Theorems
4.6--4.7]{DL09}. The number of codewords of an $\texttt{S}$-basis
of $\mathcal{B}$ is called rank of $\mathcal{B}$ and it will be
denoted as $\texttt{rank}_{\texttt{S}}(\mathcal{B}).$ We also
will denote by $\texttt{row}(A)$ the $\texttt{S}$-linear code generated by the rows of the  matrix $A.$

The set of all $k\times\ell$ matrices over $\texttt{S}$ will be
denote by $\texttt{S}^{k\times\ell}.$ A matrix
$A\in\texttt{S}^{k\times\ell}$ is said be a \emph{full-rank
matrix} if $\texttt{rank}(A)=k$. A matrix
$A$ with entries in $\texttt{S}$ is called a \emph{generator
matrix} for the code $\mathcal{B}$ if the set of rows of $A$ is a
basis of $\mathcal{B}$, therefore it is a  full-rank matrix. We will denote by
$\texttt{GL}_k(\texttt{S})$ the group of invertible matrices in
$\texttt{S}^{k\times k}$. We say that the matrices $A$ and $B$ in
$\texttt{S}^{k\times \ell}$ are row-equivalent if there exists a matrix
$P\in\texttt{GL}_k(\texttt{S})$ such that $B=PA$.

 \begin{Definition}\label{pivot}
Let $A$ be a matrix in $\texttt{S}^{k\times\ell}$ and $A[i:]$ the
$i$-th row of $A;$
  $A[:j]$ the $j$-th column of $A;$
 $A[i;j]$ the $(i,j)$-entry of $A$.
\begin{enumerate}
    \item The \emph{valuation function} of $A$ is the mapping $\vartheta_A:\{1,\cdots,k\}\rightarrow\{0,1,\cdots,s\},$
defined by
$$\vartheta_A(i):=\vartheta_\texttt{S}(A[i:]):=\texttt{min}\{\vartheta_\texttt{S}(A[i;j])\,|\, 1\leq
j\leq\ell \}.$$
   \item The \emph{pivot} of a nonzero row $A[i:]$ of
$A,$ is the first entry among all the entries least  with
valuation in that row. By convention, the pivot of the zero row is
its first entry.
    \item The \emph{pivot function} of $A$ is the mapping $\label{pivot2}{\rho}:\{1,\cdots,k\}\rightarrow\{1,\cdots,\ell\},$
defined by  $$\rho(i):= \texttt{min}\biggl\{j\in\{1;\cdots;\ell\}
\,|\, \vartheta_{S}(A[i;j])=\vartheta_i  \biggr\}.$$

\end{enumerate}

\end{Definition}

Note that from the definition the pivot of the row
$A[i:]$ is the element $\label{pivot1} {A[i,\rho(i)]}.$ Let
$\varrho$ be a ring automorphism of $\texttt{S}$, it is clear that the pivot function and
valuation function of the matrices
$A$ and $\left(\varrho(A[i;j])\right)_{\substack{1\leq i\leq k \\
1\leq j\leq\ell}}$ provide the same values.

\begin{Definition}[Matrix in row standard form \cite{CNKD}]\label{defi0}
A matrix $A\in \texttt{S}^{k\times\ell}$ is in {row standard
form} if it satisfies the following conditions
\begin{enumerate}
     \item The pivot function of $A$ is injective and the valuation function of $A$ is increasing,
    \item for all $i\in\{1,\cdots,k\},$ there is $\vartheta_i\in\{0,1,\cdots,s-1\}$ such that
    $A[i;\rho(i)]=\theta^{\vartheta_i}$ and $ {A[i:]}\in(\theta^{\vartheta_i}\texttt{S})^\ell$ and
    \item  for all pairs $i, t\in\{1,\cdots,k\}$ such that $t\neq i,$ then 
    \begin{enumerate}
      \item either $i>t$  and
     $ \texttt{deg}_\texttt{R}\left({A[t;\rho(i)]}\right)<\vartheta_{i}$,
     \item      or $A[i;\rho(t)]=0$.
    \end{enumerate}
\end{enumerate}
\end{Definition}

Note that a matrix in row
standard form is a nonzero matrix and that its rows are linearly independent. Moreover, if $A\in
\texttt{S}^{k\times\ell}$ is in row standard form, then for any
ring
automorphism $\varrho$ of $\texttt{S},$ the matrix $\left(\varrho(A[i;j])\right)_{\substack{1\leq i\leq k \\
1\leq j\leq\ell}}$ is also in row standard form.

Let  $A\in\texttt{S}^{k\times\ell}$  be a nonzero matrix,  we say that a matrix
$B\in\texttt{S}^{k\times\ell}$ is the \emph{row standard form} of
$A$ if  $B$ is in row standard form and  $B$ is row-equivalent to
$A$. A proof of the existence and unicity of the row standard form
of a matrix can be found in \cite{CNKD}. Since the set of all
generator matrices of any $\texttt{S}$-linear code $\mathcal{B}$
is a coset under   row equivalence, it follows that
$\mathcal{B}$ has a unique generator matrix in row standard
form that will be denoted by $\texttt{RSF}(\mathcal{B})$.
As usual we define the type of a linear code as follows.

\begin{Definition}[Type of a linear code]
 Let $\mathcal{B}$ be an $\texttt{S}$-linear code of length $\ell.$ Denoted by  $\theta^{\vartheta_i}$ the  $i$-th pivot of
$\texttt{RSF}(\mathcal{B}).$ The \emph{type} $\mathcal{B}$ is the
$(s+1)$-tuples
$$(\ell;k_0,k_1,\cdots,k_{s-1})$$ where $k_{\textgoth{t}}:=|\{\vartheta_i\,|\,\vartheta_i=\textgoth{t}\}|.$

\end{Definition}

Note that if $(\ell;k_0,k_1,\cdots,k_{s-1})$ is the type of an
$\texttt{S}$-linear code $\mathcal{B}$ then we can be compute the
$\texttt{S}$-rank of $\mathcal{B}$ and the number of codewords of
$\mathcal{B},$ of the following way:
$$\texttt{rank}_{\texttt{S}}(\mathcal{B})=\sum\limits_{\textgoth{t}=0}^{s-1}k_{\textgoth{t}},\text{
~~and~~
    }|\mathcal{B}|=q^{m\left(\sum\limits_{t=0}^{s-1}k_{t}(s-t)\right)}.$$

\section{Galois action on $\mathcal{L}(\texttt{S}^\ell).$}

Let $\texttt{S}|\texttt{R}$ be a  Galois extension of
finite chain ring  with Galois group $G.$ The Galois group $G$
acts on $\mathcal{L}(\texttt{S}^\ell)$ as follows; Let $\mathcal{B}$ in
$\mathcal{L}(\texttt{S}^\ell)$ and $\sigma$ in $G$
\begin{align}\sigma(\mathcal{B})=\left\{(\sigma(c_0),\sigma(c_1),\cdots,\sigma(c_{\ell-1}))\,\biggr|\,(c_0,c_1,\cdots,c_{\ell-1})\in\mathcal{B}\right\}.\end{align}

\begin{Definition}[Galois invariance]\label{G-inv}
A linear code $\mathcal{B}$ over $\texttt{S}$ is called
\emph{Galois invariant} if $\sigma(\mathcal{B}) = \mathcal{B}$
for all $\sigma\in G$.
\end{Definition}
An direct consequence of this definition is the following
fact.  Let $\mathcal{B}$ be an $\texttt{S}$-linear code  and $m$ an even number,
if $\mathcal{B}$ is a Galois invariant code then
$\mathcal{B}^{\perp_{\varphi'}}=\mathcal{B}^{\perp}$ (note that
$\mathcal{B}^{\perp_{\varphi'}}=(\sigma^{\frac{m}{2}}\left(\mathcal{B}\right))^{\perp_E}$). Therefore,
we will  consider only the euclidean inner product from now on.

\begin{Lemma}\label{in-du} Let $\mathcal{B}$ be an $\texttt{S}$-linear code  and $A$ a generator matrix of $\mathcal{B}$.
\begin{enumerate}
\item $\sigma(\mathcal{B}^{\perp})=\sigma(\mathcal{B})^{\perp},$
and $\sigma(\texttt{row}(A))=\texttt{row}(\sigma(A)),$ for all
$\sigma\in G$. \item The following assertions are equivalent:
\begin{enumerate}
    \item $\mathcal{B}$ is Galois invariant;
    \item $\mathcal{B}^{\perp}$ is Galois invariant.
    \end{enumerate}
\end{enumerate}
\end{Lemma}

The following theorem
allows us to check the Galois invariance of a code by checking its generator matrix in
row standard form.

\begin{Theorem}\label{thm0} Let $\mathcal{B}$ be an $\texttt{S}$-linear code
and $A\in\texttt{S}^{k\times\ell}$ a
generator matrix of $\mathcal{B}.$  Then the following facts are
equivalent.
\begin{enumerate}
    \item $\mathcal{B}$ is Galois invariant.
    \item $\texttt{RSF}(\mathcal{B})$ in
$\texttt{R}^{k\times\ell}.$
\end{enumerate}

\end{Theorem}

\begin{proof} $\quad $
\begin{description}
    \item[$1.\Rightarrow 2.$]  Let $\sigma$ in
$G.$ Then the matrix $\sigma(\texttt{RSF}(\mathcal{B}))$ is the
generator matrix in row standard form of $\sigma(\mathcal{B}),$ by
the uniqueness of generator matrix in row standard form, it
follows
$\texttt{RSF}(\mathcal{B})=\sigma(\texttt{RSF}(\mathcal{B}))$ for
all $\sigma$ in $G$, thus
  $\texttt{RSF}(\mathcal{B})\in\texttt{Fix}_\texttt{S}(G)^{k\times\ell}.$
  As $\texttt{S}|\texttt{R}$ is a Galois separable extension with Galois group $G,$ it follows
  that $\texttt{Fix}_\texttt{S}(G)=\texttt{R}.$ Hence $\texttt{RSF}(\mathcal{B})\in\texttt{R}^{k\times\ell}.$
    \item[$2.\Rightarrow1.$] If
$\texttt{RSF}(\mathcal{B})\in\texttt{R}^{k\times\ell},$ Then
$\sigma(\texttt{RSF}(\mathcal{B}))=\texttt{RSF}(\mathcal{B})$ is a
generator matrix of $\mathcal{B}$ and of $\sigma(\mathcal{B})$,
therefore $\mathcal{B}$ is Galois invariant (see \cite[Theorem 1]{MNR13}).
\end{description}
\end{proof}

\begin{Corollary} Let $\mathcal{B}$ be a linear code over $\texttt{S}$, $\mathcal{B}$ is Galois invariant if and
only if
$\texttt{RSF}(\mathcal{B})=\texttt{RSF}(\texttt{Res}(\mathcal{B}))$.
\end{Corollary} \noindent The proof follows directly from {Theorem}~\ref{thm0} above. We have also the following result.

\begin{Corollary}\label{type-inv} Let $\mathcal{B}$ be a linear code over $\texttt{S}$ of the type
$(\ell;k_0,k_1,\cdots,k_{s-1}).$ Then the following conditions are
equivalent.
\begin{enumerate}
    \item $\mathcal{B}$ is Galois invariant,
    \item $\texttt{Res}_\texttt{R}(\mathcal{B})$ is of type $(\ell;k_0,k_1,\cdots,k_{s-1}).$
\end{enumerate}

\end{Corollary}


 For all $\mathcal{B}_1,\, \mathcal{B}_2\in\mathcal{L}(\texttt{S}^\ell),$
$\mathcal{B}_1\vee\mathcal{B}_2=\mathcal{B}_1+\mathcal{B}_2$ is
the smallest  $\texttt{S}$-linear code  containing
$\mathcal{B}_1$ and $\mathcal{B}_2$. Note
 that
$\left(\mathcal{L}(\texttt{S}^\ell);\cap, \vee\right)$ is a
lattice and for each $\mathcal{E}\subseteq\texttt{S}^\ell$ we define  $\texttt{Ext}(\mathcal{E})$, the
\emph{extension code} of $\mathcal{E}$ to $\texttt{S}$,  as the code form by all
$\texttt{S}$-linear combinations of elements in $\mathcal{E}$.

\begin{Proposition}\label{ope} The operators
\begin{align}\mathcal{L}(\texttt{S}^\ell)\overset{\texttt{Tr}_\texttt{R}^\texttt{S};\texttt{Res}_\texttt{R}}{\underset{\texttt{Ext}}{\rightleftarrows
}}\large{\mathcal{L}}_\ell(\texttt{R})\end{align} are lattice
morphisms. Moreover,
$$\texttt{Ext}(\mathcal{C}^{\perp})=\texttt{Ext}(\mathcal{C})^{\perp}\hbox{ and }
\texttt{Tr}_\texttt{R}^\texttt{S}(\texttt{Ext}(\mathcal{C}))=\texttt{Res}_\texttt{R}(\texttt{Ext}(\mathcal{C}))=\mathcal{C}
\hbox{ for all }\mathcal{C}\in\large{\mathcal{L}}_\ell(\texttt{R}).
$$

\end{Proposition}

\begin{proof}  Let $\mathcal{B}$ and $\mathcal{B}'$ be two
$\texttt{S}$-linear codes. The trace map
$\texttt{Tr}_\texttt{R}^\texttt{S}$ is surjective and we have
$\texttt{Tr}_\texttt{R}^\texttt{S}(\mathcal{B}+\mathcal{B}')=\texttt{Tr}_\texttt{R}^\texttt{S}(\mathcal{B})+\texttt{Tr}_\texttt{R}^\texttt{S}(\mathcal{B}')$
and
$\texttt{Tr}_\texttt{R}^\texttt{S}(\mathcal{B}\cap\mathcal{B}')=\texttt{Tr}_\texttt{R}^\texttt{S}(\mathcal{B})\cap\texttt{Tr}_\texttt{R}^\texttt{S}(\mathcal{B}').$
Applying  Theorem~\ref{del} we get that $\texttt{Res}_\texttt{R}$
is a lattice morphism. On the other hand, let $\mathcal{C}$ in
$\large{\mathcal{L}}_\ell(\texttt{R}),$ the $\texttt{S}$-linear
code $\texttt{Ext}(\mathcal{C})$ is Galois invariant, thus
$\texttt{Tr}_\texttt{R}^\texttt{S}(\texttt{Ext}(\mathcal{C}))=\texttt{Res}_\texttt{R}(\texttt{Ext}(\mathcal{C}))=\mathcal{C}$
and
$\texttt{Ext}(\mathcal{C}^{\perp})=\texttt{Ext}(\mathcal{C})^{\perp}.$

\end{proof}

\begin{Definition}[Galois closure and Galois interior]\label{defi}

Let $\mathcal{B}$ be a linear code over $\texttt{S}.$
\begin{enumerate}
    \item The \emph{Galois closure} of $\mathcal{B},$ denoted  by $\widetilde{\mathcal{B}}$, is the smallest linear
code over $\texttt{S},$ containing $\mathcal{B},$ which is Galois
invariant,
$$\widetilde{\mathcal{B}}:=\bigcap\biggl\{\mathcal{T}\in\mathcal{L}(\texttt{S}^\ell)\,\biggr|\,\mathcal{T}\subseteq\mathcal{B}\text{  and  } \mathcal{T} \text{ Galois invariant }\biggr\}.$$
    \item The \emph{Galois interior} of $\mathcal{B},$ denoted $\overset{\circ
}{\mathcal{B}},$ is the greatest $\texttt{S}$-linear subcode of
$\mathcal{B},$ which is Galois invariant,
$$\overset{\circ
}{\mathcal{B}}:=
\bigvee\biggl\{\mathcal{T}\in\mathcal{L}(\texttt{S}^\ell)\,\biggr|\,\mathcal{T}\supseteq\mathcal{B}\text{
and  } \mathcal{T} \text{ Galois invariant }\biggr\}.$$
    \end{enumerate}

\end{Definition}

\noindent A map $\texttt{J}_G :\mathcal{L}(\texttt{S}^\ell) \rightarrow
\mathcal{L}(\texttt{S}^\ell)$ is called  a \emph{Galois operator}
if $\texttt{J}_G$ is an morphism of lattices such that
\begin{enumerate}
  \item $\texttt{J}_G(\texttt{J}_G(\mathcal{B}))=\texttt{J}_G(\mathcal{B})$
and
\item for all $\mathcal{B}$ in $\mathcal{L}(\texttt{S}^\ell)$ the code
$\texttt{J}_G(\mathcal{B})$ is Galois invariant.
\end{enumerate}
 The Galois
closure and Galois interior are indeed Galois operators and
$\widetilde{\overset{\circ }{\mathcal{B}}}=\overset{\circ
}{\mathcal{B}}$, $ \overset{\circ
}{\widetilde{\mathcal{B}}}=\widetilde{\mathcal{B}}.$ From
{Definition}~\ref{defi}, it follows that $\mathcal{B}$ is Galois
invariant if and only if $\widetilde{\mathcal{B}}=\overset{\circ
}{\mathcal{B}}$.

\begin{Proposition}\label{oci} If $\mathcal{B}$ is a linear code over
$\texttt{S}$ then $\overset{\circ }{\left(
\mathcal{B}^{\perp}\right)}=\left(
\widetilde{\mathcal{B}}\right)^{\perp}.$
\end{Proposition}

\begin{proof}
It is clear that $\overset{\circ }{
\mathcal{B}}\subseteq\mathcal{B}$ and, by duality,
$\left(\overset{\circ }{
\mathcal{B}}\right)^{\perp}\subseteq\mathcal{B}^{\perp}.$ Now
$\overset{\circ }{ \mathcal{B}}$ is Galois invariant therefore
$\left(\overset{\circ }{ \mathcal{B}}\right)^{\perp}$ is Galois
invariant by  {Remark}\,\ref{in-du} and it contains
$\mathcal{B}^{\perp}.$ Note that $
\widetilde{\mathcal{B}^{\perp}}$ is the smallest Galois invariant
linear code containing $\mathcal{B}^{\perp}$ hence
$\widetilde{\left(\mathcal{B}^{\perp}\right)}\subseteq\left(\overset{\circ
}{ \mathcal{B}}\right)^{\perp}.$
 Since $\mathcal{B}^{\perp}\subseteq\widetilde{\mathcal{B}^{\perp}}$, again by duality we have that $\left(\widetilde{\mathcal{B}^{\perp}}\right)^{\perp}\subseteq\mathcal{B}.$
 Now the code $\left(\widetilde{\mathcal{B}^{\perp}}\right)^{\perp}$ is
Galois invariant and is contained in $\mathcal{B}$, since the largest
code that is Galois invariant and contained in $\mathcal{B}$ is
 $\overset{\circ }{ \mathcal{B}}$, it follows that
$\left(\widetilde{\mathcal{B}^{\perp}}\right)^{\perp}\subseteq\overset{\circ
}{\mathcal{B}}$, and both inclusions give the equality.
\end{proof}

Let $\{\alpha_0,\alpha_1,\cdots,\alpha_{m-1}\}$ be a free
$\texttt{R}$-basis of $\texttt{S}$ and
$\{\alpha_0^*,\alpha_1^*,\cdots,\alpha_{m-1}^*\}$ its trace-dual
basis. We define the $i$-th projection as
\begin{align} \begin{array}{cccc}
  \texttt{Pr}_i: & \texttt{S}^{\ell} & \rightarrow & \texttt{R}^{\ell} \\
    & \emph{c} & \mapsto &
    \texttt{Tr}_{\texttt{R}}^{\texttt{S}}(\alpha^*_i\emph{c}).
\end{array}\quad i=0,\ldots,m-1.
\end{align}
Since
$\emph{c}_j=\sum\limits_{i=0}^{m-1}\texttt{Tr}_\texttt{R}^\texttt{S}(\alpha^*_i\emph{c}_j)\alpha_i,$
for all $\emph{c}_j\in\texttt{S},$ and $\mathcal{B}$ is linear
over $\texttt{S},$ it follows that
$\texttt{Pr}_i(\mathcal{B})=\texttt{Tr}_{\texttt{R}}^{\texttt{S}}(\mathcal{B}).$

\begin{Lemma}\label{proj} Let $\mathcal{B}$ be a linear code over $\texttt{S}$. Then
$\mathcal{B}\subseteq\texttt{Ext}_\texttt{S}(\texttt{Tr}_{\texttt{R}}^{\texttt{S}}(\mathcal{B}))$
and
$\texttt{Res}_\texttt{R}(\mathcal{B})\subseteq\texttt{Tr}_\texttt{R}^\texttt{S}(\mathcal{B}).$
\end{Lemma}

 The following lemma relates the Galois closure and
Galois interior with  the constructions of the trace code, the restriction
code and the extension code.

\begin{Lemma}\label{int}
Let $\mathcal{B}$ be a linear code over $\texttt{S}.$ Then
$\overset{\circ
}{\mathcal{B}}=\texttt{Ext}(\texttt{Res}_\texttt{R}(\mathcal{B}))=\underset{\sigma
\in G}{\bigcap }\sigma (\mathcal{B}).$
\end{Lemma}

\begin{proof} By {Definition}\,\ref{defi} and {Lemma}\,\ref{proj} we have
$\texttt{Ext}(\texttt{Res}_\texttt{R}(\mathcal{B}))\subseteq\overset{\circ
}{\mathcal{B}}\subseteq\mathcal{B}.$ On the other hand, the linear
code $\overset{\circ }{\mathcal{B}}$ over $\texttt{S}$  is Galois
invariant therefore we have
$\texttt{Ext}(\texttt{Res}_\texttt{R}(\overset{\circ
}{\mathcal{B}}))=\overset{\circ }{\mathcal{B}}$ (see \cite[Theorem 1]{MNR13}). Since
$\overset{\circ }{\mathcal{B}}\subseteq\mathcal{B}$ we get that
$\overset{\circ
}{\mathcal{B}}\subseteq\texttt{Ext}(\texttt{Res}_\texttt{R}(\mathcal{B})).$
 Thus we have
$\texttt{Ext}(\texttt{Res}_\texttt{R}(\mathcal{B}))=\underset{\sigma
\in G}{\bigcap }\sigma (\mathcal{B})$ (see \cite[Corollary 1]{MNR13}).

\end{proof}

\begin{Proposition}\label{clo} If $\mathcal{B}$ be a linear code over $\texttt{S}$ then
$\widetilde{\mathcal{B}}=\texttt{Ext}(\texttt{Tr}_\texttt{R}^\texttt{S}(\mathcal{B}))=\bigvee_{\sigma\in
G}\sigma(\mathcal{B}).$
\end{Proposition}

\begin{proof}
$$\begin{array}{rcl}
  \widetilde{\mathcal{B}} & = & ((\widetilde{\mathcal{B}})^{\perp})^{\perp}, \hbox{  by   \cite[Theorem 3.10 (iii)]{NS00}} \\
  &=& \left(\overset{\circ }{\left(\mathcal{B}^{\perp}\right)}\right)^{\perp},\hbox{  by    {Proposition}\,\ref{oci}},\\
 &=& \texttt{Ext}(\texttt{Res}_\texttt{R}(\left(\mathcal{B}^{\perp}\right))^{\perp}, \hbox{  by    {Lemma}\,\ref{int}} \\
   &=& \texttt{Ext}\left(\texttt{Res}_\texttt{R}\left(\mathcal{B}\right)^{\perp}\right)^{\perp},\hbox{ by    {Proposition}}\,\ref{oci},\\
   &=& \texttt{Ext}\left(\texttt{Tr}_\texttt{R}^\texttt{S}\left(\mathcal{B}\right)^{\perp}\right)^{\perp},\hbox{  by Theorem 1}.
  \end{array}
$$
Note that
$\texttt{Ext}\left(\texttt{Tr}_\texttt{R}^\texttt{S}\left(\mathcal{B}\right)^\perp\right)=\texttt{Ext}\left(\texttt{Tr}_\texttt{R}^\texttt{S}\left(\mathcal{B}\right)\right)^\perp$, therefore
it follows that
$\widetilde{\mathcal{B}}=\left(\texttt{Ext}\left(\texttt{Tr}_\texttt{R}^\texttt{S}(\mathcal{B})\right)^{\perp}\right)^{\perp}.$
Again by \cite[Theorem 3.10(iii)]{NS00} we also have that
$\widetilde{\mathcal{B}}=
\texttt{Ext}\left(\texttt{Tr}_\texttt{R}^\texttt{S}(\mathcal{B})\right).$
Since $\mathcal{B}\subseteq\bigvee_{\sigma\in
G}\sigma(\mathcal{B})$ is Galois invariant then  we have
$\widetilde{\mathcal{B}}\subseteq\bigvee_{\sigma\in
G}\sigma(\mathcal{B}).$ Finally
$\sigma(\mathcal{B})\subseteq\widetilde{\mathcal{B}}$ for all
$\sigma\in G$, therefore  $\bigvee_{\sigma\in
G}\sigma(\mathcal{B})\subseteq\widetilde{\mathcal{B}}.$
\end{proof}

\begin{Remark}\label{rem} Note that by  {Lemma}\,\ref{int} and  {Propostion}\,\ref{clo} we get the following fact.
Let $\mathcal{B}$ be an $\texttt{S}$-linear code,
then $\texttt{Res}_\texttt{R}(\overset{\circ
}{\mathcal{B}})=\texttt{Res}_\texttt{R}(\mathcal{B})$ and
$\texttt{Res}_\texttt{R}(\widetilde{\mathcal{B}})=\texttt{Tr}_\texttt{R}^\texttt{S}(\mathcal{B}).$
Thus (by Delsarte's Theorem)
$\texttt{Res}_\texttt{R}(\mathcal{B}^{\perp})=\texttt{Res}_\texttt{R}(\mathcal{B})^{\perp}$
if and only if $\mathcal{B}$ is Galois invariant.
\end{Remark}

\begin{Remark}\label{Stichtenoth1} In the  case that $S,R$ are finite fields  the properties of the Galois closure and Galois interior as well as Proposition~\ref{oci} and Lemma~\ref{int} were stated by Stichtenoth in \cite{stichtenoth}. \end{Remark}

\noindent Note also that we have
$\widetilde{\widetilde{\mathcal{B}}}=\widetilde{\mathcal{B}}$, thus
from  Remark~\ref{rem} it follows that
$$\texttt{Res}_\texttt{R}(\widetilde{\widetilde{\mathcal{B}}})=\texttt{Tr}_\texttt{R}^\texttt{S}(\widetilde{\mathcal{B}})\hbox{ and }
\texttt{Res}_\texttt{R}(\widetilde{\mathcal{B}})=\texttt{Tr}_\texttt{R}^\texttt{S}(\mathcal{B}).$$
Hence
$\texttt{Tr}_\texttt{R}^\texttt{S}(\mathcal{B})=\texttt{Tr}_\texttt{R}^\texttt{S}(\widetilde{\mathcal{B}})$
(see \cite[Proposition 1]{MNR13}). Thus as a corollary we recover
the following result.

\begin{Corollary}[Theorem~2, \cite{MNR13}]\label{ans1}
The $\texttt{S}$-linear code $\mathcal{B}$ is  Galois invariant if
and only if
$\texttt{Tr}_{\texttt{R}}^{\texttt{S}}(\mathcal{B})=\texttt{Res}(\mathcal{B}).$
\end{Corollary}

 Finally we are in condition  for enunciate a Galois correspondence statement. For any $\mathcal{B}$ in $\large{\mathcal{L}}\left(\texttt{S}^\ell\right),$ we consider $\large{\mathcal{L}}(\mathcal{B})$
 the lattice of $\texttt{S}$-linear subcode of $\mathcal{B}$. Let us define
\[\begin{array}{cccc}
  \texttt{Stab}: & \large{\mathcal{L}}(\mathcal{B}) & \rightarrow & \texttt{Sub}(G) \\
    & \mathcal{T} & \mapsto & \texttt{Stab}(\mathcal{T}),
\end{array}\textrm{ ~~~~~and~~~~~ }\begin{array}{cccc}
 \texttt{Fix}_\mathcal{B} : & \texttt{Sub}(G) & \rightarrow & \large{\mathcal{L}}(\mathcal{B}) \\
    & H & \mapsto &  \underset{\sigma \in H}{\cap }\sigma
    (\mathcal{B}),
\end{array}\]
where $\texttt{Stab}(\mathcal{T})=\left\{\sigma\in G\,\biggr|\,
\sigma(\textbf{c})=\textbf{c},\,\text{ for all
}\textbf{c}\in\mathcal{T}\right\}.$

Let $H$ a subgroup of $G,$ we say that $\mathcal{B}$ is
{$H$-invariant} if
$\texttt{Fix}_\mathcal{B}(H)=\mathcal{B}$. Note that
$\texttt{Fix}_\mathcal{B}(H)$ is an $H$-interior of $\mathcal{B}.$
 From
{Lemma}\,\ref{int} it follows that
$$\texttt{Fix}_\mathcal{B}(H)=\texttt{Ext}(\texttt{Res}_{\texttt{T}}(\mathcal{B})),$$
where $\texttt{T}=\texttt{Fix}_\texttt{S}(H).$ Moreover
$\texttt{Fix}_\mathcal{B}(\texttt{Stab}(\mathcal{B}))=\mathcal{B}\text{
~and~}\texttt{Stab}(\texttt{Fix}_\mathcal{B}(H))=H.$ Therefore we have a
Galois correspondence on $\large{\mathcal{L}}(\mathcal{B})$ as follows.

\begin{Theorem}\label{ans2} For each $\mathcal{B}$ in $\large{\mathcal{L}}\left(\texttt{S}^\ell\right),$
the pair $\left(\texttt{Stab};\texttt{Fix}_\mathcal{B}\right)$ is
a Galois correspondence between $\mathcal{B}$ and $G.$
\end{Theorem}

\section{Rank bounds}

 Let $\mathcal{B}$ be an $\texttt{S}$-linear code and $\{\textbf{c}_i\,|\,
1\leq i\leq k\}$ be the $\texttt{S}$-basis of $\mathcal{B}$  in row standard form, i.e. $\textbf{c}_i:=\texttt{RSF}(\mathcal{B})[i:].$ 
For each
$i=1,2,\ldots, k$, we will denote by $m_i$ the integer such that
$\sigma^{m_i}(\textbf{c}_i)=\textbf{c}_i$ and
$\sigma^{m_i-1}(\textbf{c}_i)\neq\textbf{c}_i.$ The set
$\{m_i\,|\,i=1,2,\cdots, k\}$ is called the \emph{level set} of $\mathcal{B}.$ 

Note that the set
$\left\{\texttt{Tr}_\texttt{R}^\texttt{S}(\alpha_j^*\textbf{c}_i)\;|\;0\leq
j<m\text{  and } 1\leq i\leq k\right\}$ is an
$\texttt{R}$-generating set of
$\texttt{Tr}_\texttt{R}^\texttt{S}(\mathcal{B})$ thus taking into account 
 {Lemma}\,\ref{proj} we have the obvious upper
bounds for the rank of restriction codes and trace codes

\begin{align}\texttt{rank}_{\texttt{R}}\left(\texttt{Res}_\texttt{R}(\mathcal{B})\right)\leq \texttt{rank}_{\texttt{S}}(\mathcal{B})\leq \texttt{rank}_{\texttt{R}}\left(\texttt{Tr}_\texttt{R}^\texttt{S}(\mathcal{B})\right)\leq m\cdot\texttt{rank}_{\texttt{S}}(\mathcal{B}).\end{align}
The inequality
$\texttt{rank}_{\texttt{R}}\left(\texttt{Res}_\texttt{R}(\mathcal{B})\right)\leq
\texttt{rank}_{\texttt{S}}(\mathcal{B})$ in (\ref{rk}) follows
from the fact that an $\texttt{R}$-basis of
$\texttt{Res}_\texttt{R}(\mathcal{B})$ is also
$\texttt{S}$-independent and
$\texttt{rank}_{\texttt{R}}(\mathcal{B})=m\texttt{rank}_{\texttt{S}}(\mathcal{B})$. Note that it is also clear that $\label{rk}\texttt{rank}_{\texttt{S}}\left(\overset{\circ}{\mathcal{B}}\right)=\texttt{rank}_{\texttt{R}}\left(\texttt{Res}_\texttt{R}(\mathcal{B})\right)$ and that $\texttt{rank}_{\texttt{R}}\left(\texttt{Tr}_\texttt{R}^\texttt{S}(\mathcal{B})\right)=\texttt{rank}_\texttt{S}\left(\widetilde{\mathcal{B}}\right)$.
We can sharpen the upper bound  in (\ref{rk}) for the rank of
trace codes as follows (note that it has some resemblances with
Shibuya's lower bound for codes over finite fields in
\cite[Theorem 1]{SMS97}).

\begin{Proposition}\label{btr} 

Let $\mathcal{B}$ be an $\texttt{S}$-linear code,
$\underline{\textbf{B}}:=\{\textbf{c}_i\,|\, 1\leq i \leq k\}$ be
the $\texttt{S}$-basis of $\mathcal{B}$ in row standard form and
$\{m_i\,|\,i=1,2,\cdots, k\}$ its level set then
\begin{equation}\label{inq0}
  \texttt{rank}_{\texttt{R}}\left(\texttt{Tr}_\texttt{R}^\texttt{S}(\mathcal{B})\right)=\texttt{rank}_\texttt{S}\left(\widetilde{\mathcal{B}}\right) \leq \sum\limits_{i=1}^{k}m_i \leq
      m\texttt{rank}_\texttt{S}\left(\mathcal{B}\right)-(m-1)\texttt{rank}_\texttt{S}\left(\overset{\circ}{\mathcal{B}}\right).
   \end{equation}

\end{Proposition}

\begin{proof} Let $\mathcal{B}'$ be the $\texttt{S}$-linear code generated by $\left\{\sigma^j(\textbf{c}_i)\;|\;0\leq j<m_i\text{  and } 1\leq
i\leq k\right\}$. It is clear that
$\widetilde{\mathcal{B}}\subseteq\mathcal{B}'$ since $\mathcal{B}'$ is Galois invariant.  Thus
$\texttt{rank}_\texttt{S}\left(\widetilde{\mathcal{B}}\right)\leq\sum\limits_{i=1}^{k}m_i$
 and  taking into account that $m_i\leq m$ we have that
\begin{eqnarray*}
  \sum\limits_{i=1}^{k}m_i & \leq  &  |\left\{\textbf{c}\in\underline{\textbf{B}}\;|\;\sigma(\textbf{c})=\textbf{c}\right\}| +\left|\left\{\sigma^j(\textbf{c})\;|\;\textbf{c}\in\underline{\textbf{B}},\; 0\leq j< m\text{  and } \sigma(\textbf{c})\neq\textbf{c}\right\}\right|\\
                           &= & \texttt{rank}_\texttt{S}\left(\overset{\circ}{\mathcal{B}}\right)+
                           m\left(\texttt{rank}_\texttt{S}\left(\mathcal{B}\right)-\texttt{rank}_\texttt{S}\left(\overset{\circ}{\mathcal{B}}\right)\right).
  \end{eqnarray*}

\end{proof}

We can  also obtain  an straight forward lower bound for the
$\texttt{R}$-rank of restriction codes as follows.
Let $\mathcal{B}$ be an $\texttt{S}$-linear code   such that $\texttt{Res}_\texttt{R}(\mathcal{B})\neq\{\textbf{0}\}$, then
$\texttt{rank}_{\texttt{R}}\left(\texttt{Res}_\texttt{R}(\mathcal{B})\right)\geq |\{i\;|\; \texttt{RSF}(\mathcal{B})[i:]\in\texttt{R}^\ell\}|$
since the rows of $\texttt{RSF}(\mathcal{B})$ which are in
$\texttt{R}^\ell$  form a matrix in row standard form, thus they are $\texttt{R}$-independent
codewords in $\texttt{Res}_\texttt{R}(\mathcal{B})$. 

A non-trivial lower bound can be found in the following result, note that it has some resemblances with Stichtenoth's lower bound for codes over finite fields in \cite[Corollary 1]{stichtenoth} since the bound is related with the rank of $\overset{\circ}{\mathcal{B^\perp}}$.

\begin{Proposition}\label{btr2} Let $\mathcal{B}$ be an $\texttt{S}$-linear code of type
$(\ell;k_0,k_1,\cdots,k_{s-1})$  and
$\{m^\perp_i\,|\,i=1,2,\cdots,\ell-k_0\}$ the level set of  $\mathcal{B}^\perp.$  Then
\begin{equation}
  \texttt{rank}_{\texttt{R}}(\texttt{Res}_\texttt{R}(\mathcal{B})) \geq 
  \ell-\sum\limits_{i=1}^{\ell-k_0}m_i^\perp
  \geq 
  mk_0-(m-1)\left(\ell-\texttt{rank}_\texttt{R}\left(\texttt{Res}_\texttt{R}\left(\mathcal{B}^\perp\right)\right)\right).
\end{equation}

\end{Proposition}

\begin{proof}
We just use  {Proposition}~\ref{btr} and Delsarte's Theorem.
\begin{eqnarray*}
  \texttt{rank}_{\texttt{R}}(\texttt{Res}_\texttt{R}(\mathcal{B})) &\geq& \ell-\texttt{rank}_{\texttt{R}}(\texttt{Res}_\texttt{R}(\mathcal{B})^\perp) \\
    &=& \ell-\texttt{rank}_{\texttt{S}}((\overset{\circ}{\mathcal{B}})^\perp),\text{  since } \overset{\circ}{\mathcal{B}}=\texttt{Ext}(\texttt{Res}_\texttt{R}(\mathcal{B})), \\
    &=& \ell-\texttt{rank}_{\texttt{S}}(\widetilde{\mathcal{B}^\perp}),\text{  since } (\overset{\circ}{\mathcal{B}})^\perp=\widetilde{(\mathcal{B}^\perp)},  \\
    &\geq & \ell-\sum\limits_{i=1}^{\ell-k_0}m_i^\perp, \text{ by Proposition\,\ref{btr}},\\
    &\geq &\ell-m\texttt{rank}_\texttt{S}\left(\mathcal{B}^\perp\right)+(m-1)\texttt{rank}_\texttt{R}\left(\texttt{Res}_\texttt{R}(\mathcal{B}^\perp)\right),\text{ By Inequality \,\ref{inq0} },\\
    & = & mk_0+(m-1)\left(\ell-\texttt{rank}_\texttt{R}\left(\texttt{Res}_\texttt{R}(\mathcal{B}^\perp)\right)\right),\text{  because }   \texttt{rank}_\texttt{S}\left(\mathcal{B}^\perp\right)=\ell-k_0,\\
    & = & mk_0+(m-1)\left(\ell-\texttt{rank}_\texttt{R}\left(\texttt{Tr}_\texttt{R}^\texttt{S}\left(\mathscr{B}\right)^\perp\right)\right),\text{ Delsarte's Theorem. }
    \end{eqnarray*}

\end{proof}

Note that the first inequality  holds because if $\mathcal C$ is an $\texttt{R}$-code of type $(\ell;k_0,k_1,\cdots,k_{s-1})$ then $\mathcal C^\perp$ is of type $(\ell;\ell-\sum_{i=0}^{s-1}k_i,k_{s-1},\cdots,k_{1})$ \cite[Theorem 3.10\,(ii)]{NS00}, in other words,  $\texttt{rank}_\texttt{R}\left(\mathcal{C}^{\perp}\right)\geq
     \ell-\texttt{rank}_\texttt{R}\left(\mathcal{C}\right)$ and the equality holds if and only if $\mathcal C$ is a free code.

From Porposition~\ref{btr} and Proposition~\ref{btr2} follows directly the following corollary relating the rank of the restriction code and the free ranks of the code and the trace code.

\begin{Corollary} Let $\mathcal B$ be an $\texttt{R}$-code of type $(\ell;k_0,k_1,\cdots,k_{s-1})$
 and $(\ell;k_0^{(t)},k_1^{(t)},\cdots,k_{s-1}^{(t)})$ 
 be
 the type of $\texttt{Tr}_\texttt{R}^\texttt{S}(\mathcal{B})$ and 
  $(\ell;k_0^{(r)},k_1^{(r)},\cdots,k_{s-1}^{(r)})$ be the type  of  $\texttt{Res}_\texttt{R}(\mathcal{B})$, then
  \begin{enumerate}
     \item  $\texttt{rank}_{\texttt{R}}(\texttt{Res}_\texttt{R}(\mathcal{B}))\geq\ell-\sum\limits_{i=1}^{\ell-k_0}m_i^\perp\geq mk_0-(m-1)k^{(t)}_0.$
    \item $mk_0 -(m-1)k^{(r)}_0 \leq\ell-k^{(t)}_0\leq m(\ell-k_0)-(m-1)(\ell-k^{(r)}_0 ).$
  \end{enumerate}
\end{Corollary}

\section{An application to Linear Cyclic Codes}

In this section we will assume  that
\textbf{$(\ell,q)=1$} and  the multiplicative order of $q$ modulo
$\ell$ will be denoted by $\texttt{ord}_\ell(q)=m$.  A subset
$\mathcal{C}$ of $\texttt{R}^\ell,$ is \emph{cyclic}, if for all
$(c_0, \cdots,c_{\ell-2}, c_{\ell-1} )\in\mathcal{C}$ we  have
$(c_{\ell-1}, c_0,  \cdots , c_{\ell-2})\in \mathcal{C}.$
 We will denote by $\mathcal{R}_\ell$  the quotient ring of
$\texttt{R}[x]$ by the ideal generated by $x^\ell-1.$ As usual, we
identify the $\texttt{R}$-modules $(\texttt{R}^\ell,+)$ and
$(\mathcal{R}_\ell,+)$ and if the polynomial $g\in\texttt{R}[x]$ has degree less or equal to $\ell-1$
 then we identify $g$ and its quotient class in
$\mathcal{R}_\ell.$ We define the map
\begin{align}\begin{array}{cccc}
 \Psi: & \texttt{R}^\ell & \rightarrow & \mathcal{R}_\ell \\
       & (c_0,c_1,\cdots,c_{\ell-1}) & \mapsto &
       c_0+c_1x+\cdots+c_{\ell-1}x^{\ell-1}+\langle x^\ell-1\rangle,
  \end{array}\end{align}
  where $\langle x^\ell-1\rangle$ is the ideal of $\texttt{R}[x]$ generated by $x^\ell-1.$ It is well known that $\Psi$ is an isomorphism of $\texttt{R}$-modules and any $\texttt{R}$-linear code $\mathcal{C}$ of
length $\ell$ is cyclic if and only if $\Psi(\mathcal{C})$ is an
ideal of $\mathcal{R}_\ell.$ The Galois  extension
$\texttt{S}$ of $\texttt{R}$ such that
$\texttt{rank}_\texttt{R}(\texttt{S})=m$ is the splitting ring of
$x^\ell-1=\prod\limits_{i=0}^{\ell-1}(x-\xi^i),$ where $\xi$ is an
element in $\Gamma(\texttt{S})$ such that $\xi^{i}\neq 1$ for
$i=0,\ldots, \ell-1$  and $\xi^\ell=1$. The  {$q$-cyclotomic
coset modulo $\ell$} containing $a$ will be denoted by $$Z_a
=\left\{aq^{j}\,\texttt{mod}\,\ell \;\biggr|\; 0\leq
j<z_a\right\},$$ where $z_a$ is the smallest nonnegative integer
such that $aq^{z_a}\equiv a\;(\texttt{mod}\;\ell).$ The set
$\texttt{Cl}_{q}(\ell):=\{a_1,a_2,\cdots,a_u\}$ will be the subset
of $\{0,1,\cdots,\ell-1\}$ such that for all
$a\in\{0,1,\cdots,\ell-1\}$ there is a unique index $i$ such that $a\in Z_{a_i}.$ Let
$a\in\texttt{Cl}_{q}(\ell),$  $\Lambda_a$ will denote the Hensel's
lift of the minimal polynomial of $\pi(\xi)^a$ over
$\mathbb{F}_{q}$ to the ring $\texttt{R},$  where $\pi(\xi)$ is a
primitive root of $x^\ell-1.$  Then
$$x^\ell-1=\prod\limits_{a\in\texttt{Cl}_{q}(\ell)}\Lambda_a,$$
is the factorization of $x^\ell-1$ into a product of distinct
basic irreducible polynomials over $\texttt{R}.$ For each element
$a\in\texttt{Cl}_{q}(\ell)$ by $\widehat{\Lambda_a}$  we will denote the
monic polynomial in $\mathcal{R}_\ell$ such that
$x^\ell-1=\widehat{\Lambda_a}\Lambda_a.$ Then there exists a pair
$(u,v)\in(\mathcal{R}_\ell)^2$ such that
$u\Lambda_a+v\widehat{\Lambda_a}=\textbf{1}$. The idempotents of
$\mathcal{R}_\ell$ are described in the following result.

\begin{Lemma}[Theorem 2.9 \cite{Wan99}] The set
$\biggl\{e_a:=v\widehat{\Lambda_a}\,\biggr|\,
a\in\texttt{Cl}_{q}(\ell)\biggr\}$ is the set of the mutually
orthogonal non-zero idempotents of $\mathcal{R}_\ell$ and
$\sum\limits_{a\in\texttt{Cl}_{q}(\ell)}e_a=1.$
\end{Lemma}

The last equality
implies the decomposition of $\mathcal{R}_\ell$ into the direct
sum of ideals of the form $\Psi\left(\mathcal{C}_a\right):=\left\langle
e_a\right\rangle$  such that
$\Psi\left(\mathcal{C}_a\right)\Psi\left(\mathcal{C}_{a'}\right)=\{\textbf{0}\}$
(since  $e_ae_{a'}=0$ if $a\neq a'$), i.e.
\begin{align}\mathcal{R}_\ell:=\underset{\substack{a\in\texttt{Cl}_{q}(\ell)}}{\bigoplus}\Psi\left(\mathcal{C}_a\right).\end{align}

Let $\mathcal{C}$ be an $\texttt{R}$-linear cyclic subcode of
$\mathcal{C}_a.$ Since $\Psi\left(\mathcal{C}_a\right)$ is a
principal ideal in $\mathcal{R}_\ell,$ there exists
$f\in\mathcal{R}_\ell$ such that
$\Psi\left(\mathcal{C}\right)=\left\langle f\right\rangle$ and
$e_{a}$ divides $f.$ If $\mathcal{C}\neq\mathcal{C}_a,$ then
$\omega_{a}\not\in\Psi\left(\mathcal{C}\right).$ Therefore there
exits an integer $\textgoth{t}\in\{0,1,\cdots,s-1\}$ such that
$f(x)=\theta^{\textgoth{t}}\omega_{a}(x)$ and  we have the following.

\begin{Proposition}\label{stru} Let $a\in\texttt{Cl}_{q}(\ell),$ the cyclic $\texttt{R}$-subcodes of the $\texttt{R}$-linear cyclic  code
$\Psi\left(\mathcal{C}_a\right):=\left\langle
\omega_{a}\right\rangle$ are
\begin{align}\{0\}\subsetneq\mathcal{C}_{a,s-1}\subsetneq\cdots\subsetneq\mathcal{C}_{a,1}\subsetneq\mathcal{C}_a,\end{align}
 and $\mathcal{C}_{\textgoth{t}_{a}}:=\theta^{\textgoth{t}_{a}}\,\mathcal{C}_a$
is the only $\texttt{R}$-linear cyclic subcode of $\mathcal{C}_a$
such that
$\theta^{s-\textgoth{t}_{a}}\mathcal{C}_{\textgoth{t}_{a}}=\{0\}$
and
$\theta^{s-\textgoth{t}_{a}-1}\mathcal{C}_{\textgoth{t}_{a}}\neq\{0\}.$
\end{Proposition}

\begin{Corollary}\label{sum} For each $\texttt{R}$-linear cyclic  code $\mathcal{C}$ of length $\ell$
there exists a unique multi-index
$$(\textgoth{t}_{a})_{a\in\texttt{Cl}_{q}(\ell)}\in\{0,1,\cdots,s\}^{\texttt{Cl}_{q
}(\ell)}$$ such that
$\mathcal{C}:=\underset{\substack{a\in\texttt{Cl}_{q
}(\ell)}}{\bigoplus}\mathcal{C}_{\textgoth{t}_{a}}$.
\end{Corollary}

 Consider the set $\circlearrowleft_\ell(\texttt{R})$ of all the
cyclic codes over $\texttt{R}$ of length $\ell$ and
$\mathcal{A}_\ell(q,s):=\{0,1,\cdots,s\}^{\texttt{Cl}_{q }(\ell)}$
the set of all the multi-indices. Corollary\,\ref{sum} establishes
that
\begin{align}\label{bij}\begin{array}{cccc}
 \Im : & \mathcal{A}_\ell(q,s) & \rightarrow & \circlearrowleft_\ell(\texttt{R}) \\
    & ~\underline{\textbf{t}}~ & \mapsto &\underset{\substack{a\in\texttt{Cl}_{q
    }(\ell)}}{\bigoplus}\mathcal{C}_{\textgoth{t}_a}.
\end{array}
\end{align}
is a bijection between the sets $\mathcal{A}_\ell(q,s)$ and
$\circlearrowleft_\ell(\texttt{R}).$ Let $\textbf{\textbf{t}}$ be
the multi-index associate to a $\texttt{R}$-linear cyclic code
$\mathcal{C}$,  the integers
$k_j=|\{a\in\texttt{Cl}_{q}(\ell)\,|\, \textgoth{t}_{a}=j\}|$ with $0\leq j\leq s-1$
determine the type
$(\ell;k_0,k_1,\cdots,k_{s-1})$ of $\mathcal{C}.$  Moreover,
$\mathcal{C}_a$ is a minimal free $\texttt{R}$-linear cyclic  code
of $\texttt{R}$-rank $z_a.$ In the rest of the paper  we will face
this question:
\begin{quote}
\textit{Let $\mathcal{C}$ be an $\texttt{R}$-linear cyclic code of
length $\ell$ How one can construct an $\texttt{S}$-linear cyclic
code $\mathcal{B}$ of length $\ell,$ such that
$\mathcal{C}=\texttt{Res}_\texttt{R}(\mathcal{B})$ and
$\mathcal{B}$ is Galois invariant? }\end{quote}

\noindent Consider the set $A:=\{a_1,a_2,\cdots,a_k\}\subseteq
\{0,1,\cdots, \ell-1\}$ and the evaluation $\texttt{ev}_\xi$ in
$\underline{\xi}:=(1,\xi,\xi^2,\cdots,\xi^{\ell-1})$ defined by
$$\begin{array}{cccc}
  \texttt{ev}_\xi: & \mathcal{P}(A) & \rightarrow & \texttt{S}^\ell \\
   & f & \mapsto &
   (f(1),f(\xi),\cdots,f(\xi^{\ell-1})).
\end{array}$$
The  $\texttt{R}$-module $\mathcal{P}(A)$ is  free and spanned by
$\{x^a\,|\,a\in A\}.$ Thus $\texttt{ev}_\xi(\mathcal{P}(A))$ is
the free $\texttt{S}$-linear code $\mathcal{B}(A)$ with generator matrix \begin{align}\label{won} W_A:=\left(%
\begin{array}{cccc}
  1 & \xi^{a_1} & \cdots & \xi^{(\ell-1)a_1} \\
  \vdots & \vdots &   & \vdots \\
  1 & \xi^{a_k} & \cdots & \xi^{(\ell-1)a_k}
\end{array}%
\right)\end{align}
and the subset $A$
of $\{0,1,\cdots,\ell-1\},$ is called \emph{defining set} of
$\mathcal{B}(A).$

 Let
$u\in\{0,1,\cdots,\ell\},$ the set of\emph{multiples} of $u$ is
 $uA:=\{ua\,\texttt{mod}\,\ell\,|\,a\in A\}$, the
\emph{opposite} of $A$ is $-A:=\{\ell-1-a\,|\,a\in A\}$ and  a
subset $A$ is said \emph{$q$-invariant } if $A=qA.$
 The
$q$-closure of $A$ is
$$\widetilde{A}:=\underset{a\in A}{\cup
}Z_a.$$ It is clear that the $q$-closure of $A$ is the smallest
$q$-invariant subset of $\{0,1,\cdots,\ell-1\}$ contained $A.$ The
\emph{complementary} of $A$ is
$\overline{A}:=\{a\in\{0,1,\cdots,\ell-1\}\,|\,a\not\in A\}.$

\begin{Proposition}\label{cyclic} Let $A$ be a subset of $\{0,1,\cdots,\ell-1\}.$ Then $\mathcal{B}(A)$ is cyclic and its generator polynomial
 is $\prod\limits_{a\in \overline{A}}(x-\xi^{-a}).$
\end{Proposition}

\begin{proof}
Consider the codeword $\textbf{c}_f =\texttt{ev}_\xi(f)=
\left(f(0); \cdots;f(\xi^{\ell-2}); f(\xi^{\ell-1})\right)$
determined by $f(x) = \sum\limits_{i=1}^k
f_ix^{a_i}\in\mathcal{P}(A).$ For $g(x) = \sum\limits_{i=1}^k
f_i\xi^{-a_i}x^{a_i}\in \mathcal{P}(A)$ and $ \left(g(0); g(\xi);
\cdots ; g(\xi^{\ell-1})\right)$ is the shift of $\textbf{c}_f$ and therefore
$\mathcal{B}(A)$ is a cyclic code. On the other hand we have
$\Psi(\textbf{c}_f)=\sum\limits_{j=1}^{\ell-1}f(\xi^j)x^j$
and
\begin{align*}
  \Psi(\textbf{c}_f)(\xi^{a}) = \sum\limits_{i=1}^kf_i\xi^{-a_i}\left(\sum\limits_{j=1}^{\ell-1}\xi^{j(a_i+a)}\right)
     \ell\sum\limits_{i=1}^kf_i\xi^{-a_i}\delta_{-a_i,a},\quad a=0,\cdots,\ell-1.
\end{align*}
Thus $a\in -\overline{A}$ if and only if
$\Psi(\textbf{c}_f)(\xi^{a})=0$ and therefore
$\Psi(\textbf{c}_f)(\xi^{a})=\left(\prod\limits_{a\in
\overline{A}}(x-\xi^{-a})\right)f(x).$ Note that $\mathcal{B}(A)$ is
an $\texttt{S}$-free module of rank $|A|$ and $\Psi$ is an
$\texttt{S}$-module isomorphism, thus $\prod\limits_{a\in
\overline{A}}(x-\xi^{-a})$ is the generator polynomial of
$\mathcal{B}(A).$

\end{proof}

It is easy to check that
$\sum\limits_{j=0}^{\ell-1}\xi^{ij}=\ell\delta_{i,0},$ for all
$i=0,1,\cdots,\ell-1$, therefore the following result holds.

\begin{Lemma}\label{dual} Let $A$ and $B$ be two subsets of  $\{0,1,\cdots,\ell-1\}.$ Then
\begin{enumerate}
    \item $A\subseteq B$ if and only if
$\mathcal{B}(A)\subseteq\mathcal{B}(B);$
    \item $A\cap(-B)=\emptyset$ if and only if $\mathcal{B}(A)\perp\mathcal{B}(B).$
\end{enumerate}

\end{Lemma}

Consider $2^{\{0,1,\cdots,\ell-1\}}$ the set of the subsets of
$\{0,1,\cdots,\ell-1\}$ and $\circlearrowleft_\ell(\texttt{S})$
the set of cyclic codes over $\texttt{S}$ of length $\ell$, from  {Lemma}\,\ref{dual} we get the following.

\begin{Corollary}\label{dual} For $~\textgoth{t}=0,1,\cdots,s,$ the map
\begin{align}
\begin{array}{cccc}
   \mathcal{B}_\textgoth{t}: & 2^{\{0,1,\cdots,\ell-1\}} & \rightarrow & \circlearrowleft_\ell(\texttt{S}) \\
    &  A & \mapsto & \theta^\textgoth{t}\mathcal{B}(A),
\end{array}
\end{align}
is a monomorphism of lattices. Moreover,
$\mathcal{B}_\textgoth{t}(A)$ decomposes as
$$\mathcal{B}_\textgoth{t}(A)=\underset{a\in\texttt{Cl}_q(\ell)}{\oplus
}\mathcal{B}_\textgoth{t}(A\cap Z_a)$$ and
$\mathcal{B}_\textgoth{t}(A)^\perp=\theta^{s-t}\mathcal{B}(-\overline{A})$
for any subset $A\subseteq\{0,1,\cdots,\ell-1\}.$
\end{Corollary}

For $\textgoth{t}=0,1,\cdots,s,$ we have
$\mathcal{B}_\textgoth{t}(\emptyset)=\{\textbf{0}\}$ and
$\mathcal{B}_\textgoth{t}(\{0,1,\cdots,\ell-1\})=(\texttt{S}\theta^{\textgoth{t}})^\ell$,
and  the following properties of the code
$\mathcal{B}_{\textgoth{t}}(A)$ hold.

\begin{Theorem}\label{main} Let $A\subseteq \{0,1,\cdots,\ell-1\}$, then $\mathcal{B}_\textgoth{t}(A)$ is Galois invariant
if and only if
$A$ is $q$-invariant.
\end{Theorem}

\begin{proof} Just note that
$\sigma(\mathcal{B}_\textgoth{t}(A))=\mathcal{B}_\textgoth{t}(qA),$
for all $\sigma\in G.$
\end{proof}

  The following result extends \cite[Theorem 5]{Bie02} to finite chain rings.

\begin{Corollary}\label{Gclos}  Let $A\subseteq \{0,1,\cdots,\ell-1\}$, then
$\mathcal{B}\left(\widetilde{A}\right)$ is the Galois closure of
$\mathcal{B}(A).$
\end{Corollary}

Consider the set $\mathcal{G}^\circ_\ell(\texttt{S})$ of all the
$\texttt{S}$-linear cyclic codes of length $\ell,$ which are
Galois invariant. Then the map
\begin{align}\begin{array}{cccc}
 \mathcal{B} : & \mathcal{A}_\ell(q,s) & \rightarrow & \mathcal{G}^\circ_\ell(\texttt{S}) \\
    & ~\underline{\textbf{t}}~ & \mapsto
    &\underset{\substack{a\in\texttt{Cl}_{q}(\ell)}}{\bigoplus}\mathcal{B}_{\underline{\textbf{t}}_a}(Z_a).
\end{array}
\end{align}
is a bijection. We consider the
$\texttt{R}$-linear cyclic code $\Im_\textgoth{t}(A)$ defined by
\begin{align}\label{cons}\Im_\textgoth{t}(A):=\texttt{Tr}_\texttt{R}^\texttt{S}(\mathcal{B}_{s-\textgoth{t}}(A))^\perp,\end{align}
where $A$ is an $q$-invariant subset of $\{0,1,\cdots,\ell-1\}.$
According to {Theorem\,\ref{main}}, $\mathcal{B}\left(A\right)$ is
Galois invariant, and by
 {Remark\,\ref{rem}} and
 {Theorem\,\ref{ans1}} we have
$\Im_\textgoth{t}(A)=\texttt{Res}_\texttt{R}(\mathcal{B}_{\textgoth{t}}(A)^\perp),$
and by  {Corollary\,\ref{dual}}
 it follows that
$\Im_\textgoth{t}(A)=\texttt{Res}_\texttt{R}(\mathcal{B}_{s-\textgoth{t}}(-\overline{A})).$
Hence
$\Im_{\textgoth{t}_{a}}(\overline{Z_a})=\texttt{Res}_\texttt{R}\left(\mathcal{B}_{s-\textgoth{t}_a}(-Z_a)\right)$
and the bijection in Equation~(\ref{bij}) can be rewritten as
\begin{align}\begin{array}{cccc}
 \Im : & \mathcal{A}_\ell(q,s) & \rightarrow & \circlearrowleft_\ell(\texttt{R}) \\
    & \underline{\textbf{t}} & \mapsto & \underset{\substack{a\in\texttt{Cl}_{q}(\ell)}}{\bigoplus}\texttt{Res}_\texttt{R}\left(\mathcal{B}_{s-\textgoth{t}_a}(-Z_a)\right).
\end{array}
\end{align}
Consider now the $\texttt{S}$-linear cyclic code given by
\begin{align*}\mathcal{B}(~\underline{\textbf{t}~}~)^\perp:=\underset{\substack{a\in\texttt{Cl}_{q}(\ell)}}{\bigoplus}\mathcal{B}_{s-\textbf{t}_a}(-Z_a),\end{align*}
from  {Proposition\,\ref{ope}},
$\Im(~\underline{\textbf{t}}~)=\texttt{Res}_\texttt{R}\left(\mathcal{B}(~\underline{\textbf{t}~}~)^\perp\right)$ and
by {Theorem\,\ref{main}(2)}
$\mathcal{B}(~\underline{\textbf{t}~}~)^\perp$ is Galois
invariant.  The following theorem gives an answer to the previous
question.
\begin{Theorem}\label{end} For each $\underline{\textbf{t}}$ in
$\mathcal{A}_\ell(q,s)$ we have that $$
\Im(~\underline{\textbf{t}}~)=\texttt{Res}_\texttt{S}\left(\mathcal{B}(~\underline{\textbf{t}~}~)^\perp\right), \quad \texttt{rank}_{\texttt{R}} (\Im(~\underline{\textbf{t}}~) )=\sum_{i=1}^u z_{a_i},$$
and $$W_{\underline{\textbf{t}}}:=\left(%
\begin{array}{c}
  \theta^{s-\textbf{t}_{a_1}}W_{a_1}  \\
  \theta^{s-\textbf{t}_{a_2}}W_{a_2} \\
  \vdots \\
  \theta^{s-\textbf{t}_{a_u}}W_{a_u} \\
\end{array}%
\right)$$ is a generator matrix of
$\mathcal{B}(~\underline{\textbf{t}~~}~)^\perp$ where $W_{a_i}$'s
are generator matrices of $\mathcal{B}(-Z_{a_i})$'s in Equation~(\ref{won}).

\end{Theorem}

{Theorem}\,\ref{end}  generalizes the construction of
\cite[Theorem 4.14]{NS00} to linear cyclic codes over finite chain
rings. Finally, it
is important to note {Theorem\,\ref{thm0}} implies that for a subset $A\subseteq\{0,1,\cdots,\ell-1\}$
 the matrix in Equation~(\ref{won}) verifies
$\texttt{RSF}(W_A)\in\texttt{R}^{|A|\times\ell}$ if and only if
$A$ is $q$-invariant.

 Finally we will show a BCH-like bound for the minimum Hamming distance ($d_H$) of this type of codes. A subset $I\subseteq\{0,1,\cdots,\ell-1\}$ is an \emph{interval} of length $v$ if
there exists $(u,w)\in\{0,1,\cdots,\ell-1\}^2$ such that
$(w,\ell)=1$ and
\begin{align}\label{interval}I=\biggl\{wu\,\texttt{mod}\,\ell;w(u+1)\,\texttt{mod}\,\ell;\cdots;w(u+v-1)\,\texttt{mod}\,\ell\biggr\}.\end{align}

\begin{Theorem}[BCH-bound] If $A$ is an interval of length
$v$ then $d_{\texttt{H}}\left(\Im_{\textgoth{t}}(A)\right)\geq
v+1.$
\end{Theorem}

\begin{proof} Let $A=\biggl\{wa_1\,\texttt{mod}\,\ell;w(a_1+1)\,\texttt{mod}\,\ell;\cdots;w(a_1+a_k-1)\,\texttt{mod}\,\ell\biggr\}$
 for some integer $w$  such
that $(w,\ell)=1.$ Then $\zeta:=\xi^w$ is also a primitive root of
$x^\ell-1.$  Suppose that $\textbf{c}$ is a nonzero codeword of
$\mathcal{B}_{s-\textgoth{t}}(-\overline{A})$ with the least
Hamming weight. Then $W_{A}\textbf{c}^T=\textbf{0}.$ Consider $
\{j\,|\,c_j\neq
0\}\subseteq\{j_1,j_2,\cdots,j_v\}:=\underline{v}.$ Consider
$\textbf{m}=(c_{j_1},c_{j_2},\cdots,c_{j_v})$ where
$\textbf{c}=(\cdots,0,c_{j_i},0,\cdots,0,c_{j_{i+1}},0,\cdots).$
Thus the equality $W_{A}\textbf{c}^T=\textbf{0}$ becomes
$W_{\underline{v}}\textbf{m}^T=\textbf{0},$ where $W_{\underline{v}}:=\left(%
\begin{array}{cccc}
  \zeta^{j_1a_1} & \cdots & \xi^{j_va_1}   \\
  \vdots &   & \vdots   \\
  \zeta^{j_1(a_1+a_v-1)} & \cdots & \xi^{j_v(a_1+a_v-1)}
\end{array}%
\right).$  We have
$$\texttt{det}(W_{\underline{v}})=-\zeta^{v\left(\sum\limits_{t=1}^vj_t\right)}\underset{1\leq
a<b\leq v}{\prod}\left(\zeta^{j_a}-\zeta^{j_b}\right)$$ is
invertible since $\zeta\in\Gamma(\texttt{S})^*.$ Therefore
$\textbf{m}=\textbf{0}$ which is a contradiction because
$\textbf{c}\neq\textbf{0}.$ Hence
$d_{\texttt{H}}\left(\Im_{\textgoth{t}}(A)\right)\geq
d_{\texttt{H}}(\mathcal{B}_{s-\textgoth{t}}(-\overline{A}))\geq
v+1.$
\end{proof}



\bibliographystyle{elsarticle-num}

\end{document}